\documentclass[conference,a4paper]{IEEEtran}

\usepackage[cmex10]{amsmath}
\usepackage{amsthm,amssymb}
\usepackage{url}

\def\ifinlinemath#1#2{#2}
\def\mathshift{$}
\catcode`$=13
\def${\ifinner\expandafter\mathshift\else\expandafter\myshift\fi}
\def\myshift#1${{\def\ifinlinemath##1##2{##1}\raisebox{0ex}[0ex][0ex]{\mathshift#1\mathshift}}}
\AtBeginDocument{\catcode`$=13}

\def\myleftparen{(}
\def\myrightparen{)}
\catcode`(=13
\catcode`)=13
\def({\ifmmode\ifinlinemath{}{\expandafter\left}\fi\myleftparen}
\def){\ifmmode\ifinlinemath{}{\expandafter\right}\fi\myrightparen}

\makeatletter
\newtheorem{ghost@theorem}{}[section]
\def\@maketheorem#1=#2;{
	\newtheorem{#1}[ghost@theorem]{#2}}
\def\maketheorem#1{
	\@for\@x:=#1\do{
		\expandafter\@maketheorem\@x;}}
\makeatother

\def\eqdef{\stackrel{\textsl{\tiny def}}{=}}

\def\N{\mathbb N}
\def\Z{\mathbb Z}
\def\R{\mathbb R}
\def\Q{\mathbb Q}

\def\abs#1{\left|#1\right|}

\def\rn#1{{\scshape #1}}

\theoremstyle{plain}
\newtheorem*{theorem*}{Theorem}
\maketheorem{%
	theorem=Theorem,%
	lemma=Lemma,%
	proposition=Proposition,%
	corollary=Corollary,%
	fact=Fact%
}
\newtheorem{claim-toplemma}{Claim}
\theoremstyle{definition}
\maketheorem{%
	definition=Definition,%
	assumption=Assumption,%
	example=Example%
}
\theoremstyle{remark}
\maketheorem{%
	observation=Observation,%
	remark=Remark,%
	claim=Claim,%
	question=Question%
}

\IEEEoverridecommandlockouts
\begin{document}

\title{On the Computing Power of $+$, $-$, and
$\times$}

\author{\IEEEauthorblockN{Marcello Mamino}
\thanks{\hskip-\parindent Date: 15$\,{\cdot}\,$IV$\,{\cdot}\,$2014}
\thanks{\hskip-\parindent%
The author has received funding from the European Research Council
under the European Community's Seventh Framework Programme (FP7/2007-2013
Grant Agreement no.~257039).}
\IEEEauthorblockA{%
Laboratoire d'Informatique de l'\'Ecole Polytechnique (\textsc{lix})\\
B\^atiment Alan Turing\\
1 rue Honor\'e d'Estienne d'Orves, Campus de l'\'Ecole Polytechnique, 91120 Palaiseau, France\\
Email: mamino@lix.polytechnique.fr}}

\maketitle
\thispagestyle{plain}
\pagestyle{plain}

\noindent\begin{abstract}\boldmath
Modify the Blum-Shub-Smale model of computation replacing the permitted
computational primitives (the real field operations) with any finite
set~$B$ of real functions semialgebraic over the rationals.
Consider the class of Boolean decision problems that
can be solved in polynomial time in the new model by machines with no machine constants.
How does this class depend on~$B$? We prove that it is always contained in
the class obtained for~$B=\{+,-,\times\}$.
Moreover, if $B$ is a set of
continuous semialgebraic functions containing~$+$
and~$-$, and such that
arbitrarily small numbers can be computed using~$B$, then we have the
following dichotomy: either our class is~$\mathsf P$ or it coincides with the
class obtained for~$B=\{+,-,\times\}$.
\end{abstract}

\def\BP{\operatorname{BP}}
\def\posslp{\text{\rm PosSLP}}
\section{Introduction}

\noindent
In this work, we study the power of computation over the real numbers to decide
classical Boolean problems. As opposed to discrete domains, there is
currently no universally accepted natural point of view on computation
over the reals, most of the existing models being roughly divided in two groups.
On the one hand, if we regard a computation over the reals as a process of
approximation to be carried out through discrete means, then we are into
the tradition of computable analysis and the bit model. On the other hand,
we can forgo some extent of realism, and consider theoretical machines
capable of directly manipulating real numbers with unbounded precision: in
this case, we are looking at models such as the Blum-Shub-Smale model and
real random access machines.
In the context of computational complexity,
adopting the second point of view means,
usually, to fix a finite basis of primitive operations that a machine can
perform on real numbers, and fixing some prescribed, often unitary, cost
for such operations: in short,
a rigorous form of counting flops.
In general, in these models, machines compute real functions of real inputs.
There is, however, a trend
to bring complexity in the Blum-Shub-Smale model, or its variants, back into
contact with classical discrete complexity, through the study of
Boolean parts: the Boolean part of a complexity class over the reals is
obtained by restricting the input and output of the corresponding machines
to Boolean values (the idea dates back to~\cite{Good94} and~\cite{Kor93},
the reader may find more information in \S22.2 of the book~\cite{BCSS},
which is also the reference for the Blum-Shub-Smale model, additional
bibliography can be found in~\cite{ABKM08}).
In this work, we will explore how the Boolean parts of
real complexity classes change by varying the set of primitive operations.
In particular, we are interested in machines performing various sets of
semialgebraic operations at unit cost.

One point of criticism to the Blum-Shub-Smale model
(raised, for example, in~\cite{Bra05}, \cite{BraCo06})
is that the only computable functions are piecewise rational.
In short: {\it why should $\sqrt{x}$ not be computable\/}?
Consider the Sum of Square Roots problem -- compare two sums of square roots
of positive integers -- which is important in computational
geometry due to ties with the Euclidean Travelling
Salesman Problem~\cite{GGJ76}. This problem is trivially solvable in
polynomial time by a real Turing machine with primitives~$+$, $-$,
and~$\sqrt{x}$ (we always assume to have equality and comparison tests), and, in fact, it can be solved in polynomial time also by
the usual real Turing machine (i.e.\ with primitives for rational
functions), but the result requires a clever argument~\cite{Tiwa92}.
Are we witnessing a coincidental fact, or is there a deeper relation between the ad-hoc
set of primitives~$\{+,-,\sqrt{x}\}$ and the one chosen by Blum, Shub, and
Smale~$\{+,-,\times,\div\}$?
As we will see, when we restrict our attention to discrete decision
problems and, say, polynomial time, adding
$\sqrt{x}$ to the
basic functions of the real Turing machine (or replacing $\times$
and~$\div$ with~$\sqrt{x}$) will not increase (or alter) its
computational power---or, more
precisely, the set of discrete decision problems that it can decide in
polynomial time. Hence, for Boolean problems,
the question we started with has an answer:
{\it the real Turing machine doesn't need
the primitive~$\sqrt{x}$, because it can simulate it}.

The study of complexity over arbitrary structures
has been initiated by Goode
in~\cite{Good94} and continued by many, see for instance Poizat's
book~\cite{Poizat} (also, in the context of recursion theory,
there has been previous work: see~\cite{Ersh81}, \cite{FrieMan92}). This line of
research focused mainly on questions such as~$\mathsf P = \mathsf{NP}$
{\it inside\/} different structures, or classes of
structures; i.e.\ considering equivalences or separations relativized to
various structures more or less in the same spirit as one relativizes to
various oracles.
Adding Boolean parts to the mix, we gain the ability to meaningfully
compare complexity
classes {\it across\/} structures.
Several problems that are complete for the Boolean part~$\BP(\mathsf P_{\R}^0)$
of the class of problems solvable in polynomial time by real Turing
machines without machine constants have been recently studied by
Allender, B\"urgisser, Kjeldgaard-Pedersen, and Miltersen~\cite{ABKM08}.
One of the $\BP(\mathsf P_{\R}^0)$-complete
problems identified
in~\cite{ABKM08}, called by them the Generic Task of Numerical Computation,
is offered as a prototype for problems that are hard for numerical, as
opposed to combinatorial, reasons; suggesting that the notion of~$\BP(\mathsf
P_{\R}^0)$-hardness may have practical value, to prove
intractability of numerical problems, much as $\mathsf{NP}$-hardness is
used for combinatorial problems. In fact, recent research adopts precisely
this point of view to assess the complexity of fixed point
problems~\cite{EtYa10}, and of semidefinite programming~\cite{TarVya07}.
We are therefore encouraged to investigate how the (analogue of) the
class~$\BP(\mathsf P_{\R}^0)$ changes when varying the computational basis,
both as a means to evaluate how generic the {\sc gtnc} really is, and as a way to
build up a toolbox of problems hard or complete for~$\BP(\mathsf
P_{\R}^0)$.

The aim of this paper is to prove that the computing power
of any finite set of real functions semialgebraic over~$\Q$
--
examples of which are the square root, a function computing
the real and imaginary parts of the roots of a seventh degree polynomial given by its coefficients, or the euclidean distance of two
ellipsoids in~$\R^3$ represented using, say, positive semidefinite
matrices
--
does not exceed the computing power of~$+$, $-$, and~$\times$.
We also prove, under reasonable
technical hypotheses, that a basis of functions
semialgebraic over~$\Q$ either solves in polynomial time precisely the
discrete problems in~$\BP(\mathsf P_{\R}^0)$, or precisely~$\mathsf P$.
For instance, to go back our little example,
the discrete problems
that the computational bases~$\{+,-,\sqrt{x}\}$, $\{+,-,\times,\div\}$,
and~$\{+,-,\times,\div,\sqrt{x}\}$
can solve in polynomial time are the same.

We will, now, spend a few words on the technical setting of our results.
Among the $\BP(\mathsf P_{\R}^0)$-complete problems identified
in~\cite{ABKM08} there is the problem $\posslp$: to decide whether a given circuit with
gates for $0,1,+,-,\times$ and no input gates represents a positive
number.
Clearly, the completeness of~$\posslp$ for~$\BP(\mathsf P_{\R}^0)$ can be
generalized to any basis~$B$ and the corresponding polynomial time class.
In other words, one can consider a
Boolean ($\subset\{0,1\}^\star$) language to be
efficiently decidable using~$B$, 
when it is decidable in polynomial time by a machine over the reals with basic
operations~$B$. Taking a different approach, one may say that a language is
efficiently decidable using~$B$ if it is polynomial time Turing reducible
to~$\posslp(B)$ -- i.e.\ $\posslp$ with gates in~$B$. These two points
of view are clearly equivalent mathematically, and, in fact, our work
can be phrased in either or both settings.
However, for the sake of clarity, we prefer to fix one and stick to it.
So, even though it may seem a less direct approach, we choose
the~$\posslp$ point of view, both because it allows finer grained
classifications -- we will state some intermediate result for
many-one instead of Turing reductions -- and
because, we believe, in total it makes the argument shorter.

For each finite set of real functions~$S$ semialgebraic over~$\Q$, we prove that $\posslp(S)$ is polynomial
time Turing reducible to~$\posslp$---this is a direct generalization of
a result in~\cite{ABKM08} proving $\BP(\mathsf P_{\R}^0)=\BP(\mathsf
P_{\R}^{\text{algebraic}})$, however we obtain our result
with different techniques, involving algebraic number theory and model
theory. Then, under the additional
hypothesis that all the functions in~$S$ are continuous, that $+$ and~$-$
are in~$S$, and that
arbitrarily small numbers can be represented by circuits with gates
in~$S$, we obtain the following dichotomy for the computational complexity
of~$\posslp(S)$. Either all the functions in~$S$ are piecewise linear, and
in this case $\posslp(S)$ is in~$\mathsf P$, or not, and in this case $\posslp(S)$
is polynomial time equivalent to~$\posslp$ (in the sense of Turing
reductions).

Finally, as a possible indication for future research,
we would like to raise the question of
machine constants (which are just $0$-ary primitives) and other
sets of primitives not semialgebraic over~$\Q$,
most importantly those that are commonly met in practice:
for instance, the typical pocket calculator functions $\sin(x)$,
$\log(x)$, $e^x$, \&c.\ (part of the arguments in this work apply to all
functions definable in an o-minimal structure over~$\R$, on the other
hand the unrestricted $\sin$~function combined with algebraic operations
easily gives rise to a problem hard for $\#\mathsf P$ via results
on~BitSLP in~\cite{ABKM08}).
Is it possible to show equivalence or separation
results involving transcendental functions?

\section{Preliminaries \& Notations}

\def\size#1{\left\|#1\right\|}
\noindent We will consider circuits whose gates operate on real
numbers (real circuits, for short).
Our circuits will have any number of input gates and precisely one
output gate, hence, for us, circuits compute multivariate real functions.
If a circuit has no input gates, we will call it a {\it closed\/} circuit:
closed circuits represent a well defined real value.
We measure the size of a circuit by the number of its gates. The depth of
a gate is the length of the longest directed path leading to it.
A {\it basis\/} is a finite set of real functions, which we intend to
use as gates. Given a basis~$B$, a $B$-circuit is a circuit with gates belonging to~$B$,
and $V(B)$ is the set of the values of all closed $B$-circuits.
We will consistently employ the same symbol to denote a circuit and the
function it represents. We will identify algebraic formul\ae\ with
tree-like circuits. The notation $\size{\,\cdot\,}$ denotes
the circuit size, while $\abs{\,\cdot\,}$ is the absolute value.

Broadly speaking,
we are interested in the efficient evaluation of the sign of closed
circuits in some basis~$B$, by means of an oracle for the evaluation of the
sign of closed circuits in some other basis~$B'$. In general, we will
employ the technique, common in computational geometry, of combining
approximate evaluation with explicit zero bounds: see~\cite{LPY05} for a
survey.

\begin{definition}[zero bound]
Let~$\mathcal{C}$ be a class of closed real circuits. We say that
$Z\colon\mathcal{C}\to\R^{>0}$ is a zero bound for~$\mathcal{C}$ if, for
all $c\in\mathcal{C}$, either $c$ evaluates to zero ($c=0$), or
$Z(c)<\abs{c}$.
\end{definition}

It is clear that, given a zero bound~$Z$ for~$\mathcal{C}$, we can decide
the sign of a circuit~$c\in\mathcal{C}$ by looking at an
approximation~$c'$ of~$c$ up to an additive error bounded by~$Z(c)/2$. In
fact, if $\abs{c'}\le Z(c)/2$, then $c=0$, otherwise $c$ and~$c'$ have the
same sign. Both directions of our argument will follow this general
recipe. We will now summarize a few facts about semialgebraic sets and Weil heights,
that we need in order to provide the ingredients.

A subset of~$\R^n$ is semialgebraic over a subring~$A$ of~$\R$ if it can
be described by a finite Boolean combination of subsets of~$\R^n$
defined by polynomial equalities~$P_i(x_1\dotsc x_n)=0$ or
inequalities~$Q_j(x_1\dotsc x_n)>0$,
with~$P_i,Q_j\in A[x_1\dotsc x_n]$.
A function is said to be semialgebraic over~$A$ if its graph, as a set, is
semialgebraic over~$A$.
As a general reference for the reader, we suggest the book of Van Den
Dries~\cite{VanDenDries}.
In this paper, we are mainly interested in functions and sets
semialgebraic over~$\Q$. Let us recall the central property of
semialgebraic sets.

Semialgebraic sets enjoy a number of properties collectively defined as
\textit{tame topology}.
Of them, it may be useful to remind that
semialgebraic sets have finitely many connected components, and
semialgebraic functions are almost everywhere infinitely differentiable.
In particular, our zero bound will be based on the following fact
(see~\cite[Chapter~2(3.7)]{VanDenDries}).

\begin{fact}\label{th-poly-bound}
If $g\colon\R\to\R$ is semialgebraic (over~$\R$), then there are $d\in\N$
and~$M>0$ such that $\abs{g(x)}\le x^d$ for all~$x>M$.
\end{fact}

In a sense, semialgebraic objects are easier to construct than it might seem
at first sight, because of the following quantifier elimination theorem.

\begin{fact}[Tarski-Seidenberg]
Let $\phi$ be a first-order formula in the field
language~$(0,1,+,-,\times)$, then there is a quantifier free
formula~$\psi$ in the same language such that
\[
	\R\models\phi\equiv\psi
\]
\end{fact}

In other words, a set is first-order definable over~$A$ in the real field,
if and only if it is semialgebraic over~$A$.
Hence, for instance, we can use constructions involving $\sup$ or~$\inf$, and
classical $\epsilon\delta$~definitions.

The Tarski-Seidenberg theorem is effective,
i.e.\ a computable procedure to obtain $\psi$ from~$\phi$
does exist, however no fast procedure is known. For our
application, the mere existence of~$\psi$ will suffice. In fact, in the
algorithms that we are going to describe, the quantifier elimination
theorem is going to be applied to finitely many formul\ae~$\phi_i$ which
are known a priori, and, in this situation, the corresponding~$\psi_i$ can
be simply hard-coded into the algorithm.

Our second ingredient is the absolute Weil height,
which was introduced by Andr\'e Weil in the context of Diophantine
geometry.
For our purpose, absolute heights are real numbers associated to points
in~$\mathbb{P}^n(\Q^{\text{alg}})$. The absolute height~$H(p)$
of~$p\in\mathbb{P}^n(\Q^{\text{alg}})$ is a positive real number meant to
represent a notion of size of~$p$. For instance, if $p$ is
in~$\mathbb{P}^n(\Q)$, then its absolute height can be determined as follows:
take a tuple~$q$ of $n+1$ coprime integers representing~$p$, then
$H(p)=\max_iq_i$. For the general definition, which is too technical for
this introduction, we refer the reader
to~\cite[Chapter~3]{Lang}.
We will summarize below the facts that we need.
The absolute height~$H(x)$ of an algebraic
number~$x$ is defined as the height
of~$(1,x)\in\mathbb{P}^1(\Q^{\text{alg}})$.
The following facts will be used to bound the result of algebraic computations.

\begin{fact}[{\cite[Property~3.3]{Waldschmidt}}]\label{th-heights-arith}
Let $a$ and~$b$ be algebraic numbers, then
\begin{align*}
H(ab) &\le H(a)H(b)
& H(a\pm b) &\le2H(a)H(b)
\end{align*}
\end{fact}

\begin{fact}\label{th-heights-alg}
Let
\[p(x) = a_0 + a_1 x + \dotsb + a_d x^d\]
be a polynomial with algebraic coefficients. Let $\zeta$ be a root
of~$p$. Then
\[
H(\zeta) \le 2^d \prod_i H(a_i)
\]
\end{fact}
\begin{proof}
Follows from~\cite[Chapter~\rn{viii} Theorem~5.9]{Silverman} observing
that $H([a_0\dotsc a_d])\le\prod_i H(a_i)$.
\end{proof}

\begin{fact}\label{th-heights-obv}
Let $\alpha\neq 0$ be an algebraic number of degree~$d$. Then
\[
H(\alpha)^{-d} \le \abs{\alpha} \le H(\alpha)^d
\]
\end{fact}
\begin{proof}
Immediate from the definition.
\end{proof}

\section{Statement of the Results} \label{sect-statement}

\noindent Now we state our main results. A brief discussion of the
hypothesis of Theorem~\ref{th-main-b}, as well as a third result which may
be of interest in certain cases, can be found in
Section~\ref{sect-addenda}. The next three sections will be devoted to
proving Theorem~\ref{th-main-a} and Theorem~\ref{th-main-b}.

\begin{definition}[$\posslp$]
Let $B=\{f_1\dotsc f_m\}$ be a finite set of functions
$f_i\colon\R^{n_i}\to\R$. The decision problem
$\posslp(B)$ is defined as follows.\\
{\bf Input}: a closed $B$-circuit~$c$\\
{\bf Output}: YES if $c>0$, NO otherwise\\
Keeping the same notation as~\cite{ABKM08}, we will
denote $\posslp(0,1,+,-,\times)$ simply as~$\posslp$.
\end{definition}

\begin{theorem}\label{th-main-a}
Let $B$ be a finite set of real functions semialgebraic over~$\Q$.
Then $\posslp(B)$ is polynomial time Turing reducible
to~$\posslp$.
\end{theorem}

\begin{theorem}\label{th-main-b}
Let $B\supset\{+,-\}$ be a finite set of continuous functions, semialgebraic over~$\Q$, such that $V(B)$ is dense. The following dichotomy holds:
either $\posslp(B)$ is in~$\mathsf P$, if all the functions in~$B$ are piecewise linear; or, if not, $\posslp$ and~$\posslp(B)$ are mutually
polynomial time Turing reducible.
\end{theorem}

\section{Circuits with Gates for Polynomial Roots} \label{sect-polynomial-roots}

\def\r{{\operatorname{r}\mskip-1mu}}
\def\ch{\operatorname{ch}}
\noindent In this section, we will study circuits in the basis
$B_d\eqdef\{0,1,+,-,\times,\ch,\r_1\dotsc\r_d\}$ where 
$\ch(x,n,z,p)$ is a choice gate
\[
\ch(x,n,z,p) =
\begin{cases}
n &\text{if $x<0$} \\
z &\text{if $x=0$} \\
p &\text{if $x>0$}
\end{cases}
\]
and $\r_d\colon \R^{d+1}\to\R$ denotes
the function mapping a tuple $(a_0\dotsc a_d)$ to the largest real root
of the polynomial~$\sum_i a_ix^i\in\R[x]$, or to~$0$ when it doesn't exist.
Observe that, for
$d_1<d_2$, a $\r_{d_2}$~gate can simulate a $\r_{d_1}$ gate, and, in
particular, a division gate
\begin{align*}
	\r_{d_1}(a_0\dotsc a_{d_1}) &= \r_{d_2}(a_0\dotsc a_{d_1},0\dotsc0)
	&\frac{x}{y} &= \r_1(-x,y)
\end{align*}
with the convention that $x/0=0$. Nevertheless, for technical reasons
which will become clear later on, we prefer to include
all the gates~$\r_1\dotsc\r_d$ in~$B_d$ as different entities.
Clearly we have
\def\mored{\le^{\mathrm P}_{\mathrm m}}
\def\moequi{\equiv^{\mathrm P}_{\mathrm m}}
\[
\posslp(0,1,+,-,\times,\ch) \moequi \posslp(B_1)
\]
and
\[
\posslp(B_1) \mored \posslp(B_2) \mored \posslp(B_3) \mored \dotsb
\]
where $\mored$ and~$\moequi$ denote polynomial time many-one reducibility
and mutual many-one reducibility---the first equivalence is
standard: keep rationals as pairs numerator-denominator.

It is quite clear that $\posslp(B_1)$ is polynomial time
Turing equivalent to~$\posslp$. With some additional effort, one can see
that $\posslp(B_1)$ is, indeed, complete, in the sense of polynomial time
many-one reductions, for~$\BP({\mathsf P}_\R^0)={\mathsf P}^\posslp$. In
fact, given a real Turing machine, one can build in polynomial time
a $B_1$-circuit that represents its computation table,
using choice gates to handle the finite transition table.
In this section, we will frequently use choice gates to simulate Boolean
circuits, which is quite easy to do in full generality.
Moreover, with choice gates, we can build circuit representations of functions involving
definitions by cases, as long as the cases are distinguished by Boolean
combinations of equalities and inequalities of other representable
functions. An example having a $B_1$-circuit representation
is, for instance, $\max(x,y)$.
A less obvious one is the
function~$f$ mapping $(a_0\dotsc a_d)$ to the number of real roots
of~$\sum_i a_ix^i$, for a fixed~$d$. We may convince ourselves that $f$ is, in fact,
representable, observing that it must take one of the $d+1$ values~$0\dotsc
d$, and, by the Tarski-Seidenberg quantifier
elimination theorem, the choice is governed by Boolean combinations of
polynomial conditions.

The goal of this section is to prove the following statement.

\begin{proposition}\label{th-polynomial-roots}
For any fixed~$d$, the decision problem
$\posslp(B_d)$ is polynomial time
Turing reducible to~$\posslp$.
\end{proposition}

It is convenient to isolate a number of intermediate steps. First we will
prove a zero bound for $B_d$-circuits in Lemma~\ref{th-heights-bound}.
Then, we will prove, in Lemma~\ref{th-perturbations},
that, for a subclass of $B_d$-circuits that we call
regular, there is an effective bound connecting the error of an approximate
evaluating procedure at each gate, with the error accumulated at the end
of the evaluation. Third, we will show how to convert $B_d$-circuits
into regular $B_d$-circuits effectively in Lemma~\ref{th-regularization}.
Finally we will prove Proposition~\ref{th-polynomial-roots} giving an
evaluation procedure for regular $B_d$-circuits based on Newton's method.

\begin{lemma}\label{th-heights-bound}
For any fixed~$d$ there is a constant~$C_d$ such that for any
closed $B_d$-circuit~$c\neq0$ we have
\[
2^{-2^{C_d\size c}} < \abs c < 2^{2^{C_d\size c}}
\]
\end{lemma}
\begin{proof}
Using Fact~\ref{th-heights-arith} and Fact~\ref{th-heights-alg} we get a
constant~$K_d$ such that
\[
H(c) < 2^{2^{K_d\size c}}
\]
where $H(c)$ denotes the absolute Weil height of the algebraic number
represented by~$c$. The lemma follows immediately 
from Fact~\ref{th-heights-obv} observing that the
degree of~$c$ is bounded by a single exponential in~$\size c$.
\end{proof}

\begin{definition}[regular circuit]\label{def-regular}
Let $c$ be a closed $B_d$-circuit. We say that a gate~$g$ of~$c$ is
regular if
\begin{enumerate}
\item $g$ is $0$ or~$1$
\item $g$ is $+$, $-$, or~$\times$ and its inputs are regular gates
\item $g$ is $\ch(x,n,z,p)$ with $x$ regular $<0$ and $n$ regular
\item $g$ is $\ch(x,n,z,p)$ with $x$ regular $>0$ and $p$ regular
\item $g$ is $\r_\delta(a_0\dotsc a_\delta)$, its inputs $a_0\dotsc
a_\delta$ are
regular, the largest real root~$\zeta$ of~$p(x)\eqdef\sum_i a_ix^i$
exists, and the first $\delta$ derivatives of~$p$ do not vanish
at~$\zeta$.
\end{enumerate}
We say that $c$ is regular if its output gate is.
\end{definition}

\begin{observation}\label{th-taint}
A regular circuit may have non~regular gates in it, because cases 3 and~4
allow some of the inputs to be non~regular.
However, the outputs of non~regular gates
do not affect the value of the circuit, in the following sense. If we
perturb the evaluation of a regular circuit
introducing any error at a non~regular gate, the result of the evaluation
will not be
affected, because the consequences of the error can not spread further than the
taint of non-regularity.
\end{observation}

\begin{lemma}\label{th-perturbations}
For any fixed~$d$ there is a constant~$E_d$ such that, for any
regular $B_d$-circuit~$c$ the following holds.
Evaluate $c$ with infinite precision, and perturb the
procedure adding at each $\r_\delta$~gate an error that is smaller in
absolute value than
\[e 2^{-2^{E_d\size c}}\]
for $e\in[0,1]$ at regular gates, and unconstrained at non-regular gates.
Then the error accumulated on the result is less than
\[
e 2^{-2^{C_d\size c}}
\]
\end{lemma}
\begin{proof}
We study the {\it loss of precision\/} incurred by our perturbed
evaluation procedure at each gate.

By Lemma~\ref{th-heights-bound} and induction on the size of~$c$, the
choices performed by regular $\ch$~gates are unaffected by the accumulated
errors.
Therefore, as a consequence of Observation~\ref{th-taint},
the errors occurring at non-regular gates
have no influence on the value of~$c$, hence we
can simply ignore non-regular gates.
We will show that, for each $g\in\{+,-,\r_1\dotsc\r_d\}$, there is a
constant $K_g$, depending only on~$g$, such that the following holds\\
($\star_g$) for $B_d$-circuits of size bounded by~$N$, if the inputs of 
a regular $g$~gate are perturbed by at most~$\epsilon<2^{-2^{K_gN}}$,
then the perturbation resulting on the output of that gate
is bounded by~$\epsilon2^{2^{K_gN}}$.

The statements ($\star_{\pm}$) are immediate with~$K_{+}=K_{-}=0$. In
order to obtain ($\star_{\times}$), it suffices to choose
$K_{\times}=C_d+1$ and the bound follows from
Lemma~\ref{th-heights-bound}. The only case left is that of
$\r_\delta$~gates.

Now we set out to find~$K_{\r_\delta}$ for a fixed~$\delta\le d$.
A regular $\r_\delta$~gate
computes the largest real root of a polynomial of
degree~$\delta$, which, by the regularity, must be a single root.
Define the set $R\subset\R^{\delta+1}$ as the set of those tuples~$x$ such
that the largest real root of the polynomial represented by~$x$ exists and
is a single root.
Consider a tuple~$x_0\in R$ representing a polynomial with
largest real root~$r= \r_\delta(x_0)$. Clearly, there is a
positive~$A\le1$
such that $\r_\delta$ is $1/A$-Lipschitz in a
neighbourhood of~$x_0$ of radius~$A$ with respect to the maximum norm,
i.e.\ for all tuples~$\epsilon_1,\epsilon_2\in\R^{\delta+1}$ with
$\abs{\epsilon_1}_\infty,\abs{\epsilon_2}_\infty<A$ we have
\[
\abs{\r_\delta(x_0+\epsilon_1)-\r_\delta(x_0+\epsilon_2)} \le
\frac{\abs{\epsilon_1-\epsilon_2}_\infty}{A}
\]
and, in particular, $R$ is an open set.
Now we intend to choose $A$ in a uniform way, to this aim, for $x\in R$,
our choice will be
\[
\alpha(x) = \sup \{\rho\le1\;|\;
\text{$\r_\delta$ is $1/\rho$-Lipschitz in $\operatorname{B}^\infty_\rho(x)$}
\}
\]
where $\operatorname{B}^\infty_\rho(x)$ denotes the open ball of
radius~$\rho$ around~$x$ in the maximum norm. By Tarski-Seidenberg
$\alpha$ is a semialgebraic function from~$R$ to~$]0,1]$ -- in fact, the
condition~$\rho\le1$ is there just to ensure that the supremum always
exists. We claim that the supremum is, indeed, a maximum,
and that $\alpha$ is $1$-Lipschitz with
respect to the maximum norm. The first claim is immediate because the
Lipschitz condition is closed. For the second claim, let $x_0$
and~$x_1$ be elements of~$R$ with $\abs{x_0-x_1}_\infty<\alpha(x_0)$.
Consider $A\eqdef\alpha(x_0)-\abs{x_0-x_1}$. Then, since
$A\le\alpha(x_0)$, we have that $\r_\delta$ is
$1/A$-Lipschitz in $\operatorname{B}_{\alpha(x_0)}(x_0)$, and, in particular, it is
$1/A$-Lipschitz in $\operatorname{B}_A(x_1)$. As a consequence
$\alpha(x_1)\ge A = \alpha(x_0) - \abs{x_0-x_1}_\infty$. Swapping $x_0$
and~$x_1$ we have the opposite inequality, hence the claim.

We turn our attention to the set~$R(N)$ of all elements of $R$ whose
coordinates can be represented
by $B_d$-circuits of size at most~$N$.
A family $\{X(t)\}_{t\in\R}$ of semialgebraic subsets of $\R^n$ is said to
be uniform if, as a subset of $\R^{n+1}$, the family is semialgebraic.
Our aim is to find a positive integer~$L$
and a uniform family of compact semialgebraic sets~$R'(t)$ such that
\[
R(N) \subset R'(2^{2^{LN}}) \subset R
\]
holds for all~$N$.
By~\cite[Theorem~2.4.1]{PreDel} there are finitely
many polynomials over the integers~$p_{i,j}$ such that $R$ can be written
in the following form
\[
R = \bigcup_i \bigcap_j \left\{x\in\R^{\delta+1} \;\middle|\; 0 < p_{i,j}(x)\right\}
\]
Let $S$ be the maximum size of the circuits representing the
polynomials~$p_{i,j}$. Fix $L$ in such a way that $C_d((d+1)N+S)\le LN$ for all positive integer~$N$.
By Lemma~\ref{th-heights-bound}
\[
R(N) \subset
\left\{ x\in\R^{\delta+1} \;\middle|\;
2^{-2^{LN}} \le \max_i \min_j p_{i,j}(x)
\right\} 
\subset R
\]
The middle set is closed, and, again by
lemma~\ref{th-heights-bound}, we can intersect it with a box of
radius~$2^{2^{LN}}$. Hence we have our set
\[
R'(t) \eqdef
\left\{ x\in\R^{\delta+1} \;\middle|\;
(\abs{x}_0 \le t) \wedge (\frac{1}{t} \le \max_i \min_j p_{i,j}(x))
\right\}
\]

Finally we consider the function
\[
\beta(t) \eqdef \max_{x\in R'(t)}\frac{1}{\alpha(x)}
\]
which, again, is semialgebraic by Tarski-Seidenberg.
The set~$R'(t)$ is compact and $1/\alpha(x)$ is continuous (since $\alpha(x)$ is
positive and Lipschitz), therefore the function~$\beta$ is well defined.
Since the structure~$(\R,+,\times)$
is polynomially bounded -- Fact~\ref{th-poly-bound} -- there are positive integers $m$
and~$n$ such that~$\beta(t) < mt^{n}$ for all~$t\ge 1$.
Let $K_{\r_\delta}$ be such that
\[
m 2^{n 2^{LN}} \le 2^{2^{K_{\r_\delta}N}}
\]
for all positive integers~$N$.
Considering circuits of size bounded by~$N$, we know that, by construction,
if we perturb the inputs of a regular~$\r_\delta$ gate by an additive
error bounded by
\[\epsilon<2^{2^{K_{\r_\delta}N}}\le\frac{1}{\beta(2^{2^{LN}})}\]
then the perturbation resulting on the output is bounded
by
\[\epsilon\beta(2^{2^{LN}})\le\epsilon2^{2^{K_{\r_\delta}N}}\]
Hence ($\star_{\r_\delta}$) is established.

We have completed the proof of the statements~($\star_g$).
Let $K$ be $\max_g(K_g)$.
Now, we choose~$E_d$ in such a way that
\[
(1 + 2^{2^{K N}})^{N-1} 2^{2^{C_d N}}
\le 2^{2^{E_d N}}
\]
for all positive integer~$N$.
We can see, by induction, that the accumulated error on a gate at
depth~$i$ is bounded by
\[
(1+2^{2^{K\size{c}}})^{i+1} 2^{-2^{E_d\size{c}}}
\]
The lemma follows.
\end{proof}

\begin{lemma}\label{th-regularization}
There is a polynomial time procedure that given a closed $B_d$-circuit~$c$
generates a regular $B_d$-circuit~$\bar{c}$ such that~$\bar{c}=c$.
\end{lemma}
\begin{proof}
Our procedure consists of two steps: first we regularize the
$\r_\delta$~gates, and then we deal with the choice gates.
We say that a gate~$g$ is quasi-regular if
\begin{enumerate}
\item $g$ is $0$ or~$1$
\item $g$ is $+$, $-$, or~$\times$ and its inputs are quasi-regular gates
\item $g$ is $\ch(x,n,z,p)$ with $x$ quasi-regular $<0$ and $n$ quasi-regular
\item $g$ is $\ch(x,n,z,p)$ with $x$ quasi-regular $=0$ and $z$ quasi-regular
\item $g$ is $\ch(x,n,z,p)$ with $x$ quasi-regular $>0$ and $p$ quasi-regular
\item $g$ is $\r_\delta(a_0\dotsc a_\delta)$, its inputs $a_0\dotsc a_\delta$ are
quasi-regular, the largest real root~$\zeta$ of~$p(x)\eqdef\sum_i a_ix^i$
exists, and the first $\delta$ derivatives of~$p$ do not vanish
at~$\zeta$.
\end{enumerate}
A quasi-regular circuit is one whose output gate is quasi-regular.
In other words, being quasi-regular is like being regular, except that the
first arguments of choice gates may be zero.

For the first step, our procedure generates a quasi-regular $B_d$-circuit~$c'$
such that~$c'=c$.
It suffices to prove that we can replace $\r_\delta(a_0\dotsc a_d)$
with a suitable circuit~$s_\delta(a_0\dotsc a_d)$ in such a way that the
output of~$s_\delta$ is
quasi-regular whenever its inputs are quasi-regular.
In order to construct~$s_\delta$, we produce a more general family of
circuits~$s_{\delta,i}$ such that $s_{\delta,i}(a_0\dotsc a_\delta)$ evaluates to the $i$-th
largest real root of~$p$ (counted with multiplicity), or to~$0$ if said root does not
exist. Clearly the choice~$s_\delta=s_{\delta,1}$ works. By induction we will show
how to build~$s_{\delta,i}$ using all~$s_{\delta',i'}$
for~$\delta'<\delta$, or for $\delta'=\delta$
and~$i'<i$. If~$i=1$, we use choice gates and polynomial conditions to
test whether the conditions for regularity fail---i.e.\ whether the degree
of~$p$ is less than~$\delta$, or, using Tarski-Seidenberg,
whether the largest root of~$p$ either does not exist or it is a root of some derivative of~$p$.
In each of these cases, we choose to use the appropriate~$s_{\delta',i'}$ or the
constant~$0$. Otherwise, we use~$\r_\delta$.
If~$i>1$, we use Tarski-Seidenberg again, to guard against the case in which
the required root does not exist, and if it exists we use $s_{\delta-1,i-1}$ applied to the coefficients
of~$p(x)/(x-s_{\delta,1}(a_0\dotsc a_\delta))$ which can be computed by polynomial
division (hence using $+$, $-$, and~$\times$, because the denominator is
a monic polynomial).

Now we have a quasi-regular circuit~$c'$ of size bounded by a multiple
(depending on~$d$) of the size of~$c$. By Lemma~\ref{th-heights-bound}, for
each gate $g$ of~$c'$, either $g=0$, or $\abs{g}>2^{-2^{C_d\size{c'}}}$.
Hence
\begin{multline*}
\ch(g,n,z,p) =\\
\ch(2^{2^{C_d\size{c'}}}g+1,n,0,\ch(2^{2^{C_d\size{c'}}}g-1,z,0,p))
\end{multline*}
where, clearly, the first arguments of the $\ch$~gates on the right hand side
can never be zero. Therefore, performing the substitution above on all
gates of~$c'$ gives us a regular circuit~$\bar{c}$. Finally, we observe that our
substitution does entail at most a linear increase in size from~$c'$
to~$\bar{c}$. In fact, $2^{2^n}$ can be computed by a
circuit of size linear in~$n$ by iterated squaring.
\end{proof}

As a last step before the proof of Proposition~\ref{th-polynomial-roots}, we
single out two statements of a technical but otherwise uncomplicated nature.
For the first one, let us remind Smale's notion of approximate zero of a real
function.

\def\piripi{following~\cite[Chapter~8 \S~1]{BCSS}}
\begin{definition}[approximate zero---\piripi]\label{def-approx-zero}
Let $f\colon\R\to\R$ be a differentiable function and $\zeta$ be a zero
of~$f$.
We say that a real
number $z$ is an approximate zero of~$f$ associated to~$\zeta$ if
the sequence~$\{z_i\}_i$ of the iterates
of Newton's approximation method applied to~$f$ starting from~$z_0=z$
satisfies the following condition for all~$i$
\[
\abs{z_i-\zeta} \le 2^{1-2^i} \abs{z-\zeta}
\]
\end{definition}

\begin{lemma}\label{th-approx-zero}
Let $f\colon\R\to\R$ be a twice differentiable function and $\zeta$ be a zero
of~$f$. Consider three real numbers $a<b<c$
such that the following conditions hold
\begin{enumerate}
\item $c-b \le (b-a)/2$
\item $b\le\zeta\le c$
\item the first derivative of~$f$ is positive on~$]a,c]$
\item \label{cond-bisection} there is $e$ such that $2^e\le f''(x)\le 2^{e+1}$ for all~$x\in[a,c]$
\end{enumerate}
then $c$ is an approximate zero of~$f$ associated to~$\zeta$.
\end{lemma}
\begin{proof}
Let $z_i$ be the $i$-th iterate of Newton's method starting from~$z_0=c$.
It is well known
(for instance \cite[Formula~2.2.2]{Atkinson}) that
\[
z_{i+1}-\zeta = \frac{f''(\xi_i)}{2f'(z_i)}(z_i-\zeta)^2
\]
for some~$\xi_i\in[\zeta,z_i]$.
Observing that
\begin{align*}
f''(\xi_i)&\le 2^{e+1} & f'(z_i)&\ge (z_i-a)2^e
\end{align*}
we get by induction on~$i$
\[
z_{i+1}-\zeta \le \frac{(z_i-\zeta)^2}{z_i - a}
\le \frac{2^{2-2^{i+1}}(c-\zeta)^2}{z_i-a} \le 2^{1-2^{i+1}}(c-\zeta)
\]
where the last inequality follows from
\[\frac{c-\zeta}{z_i-a}\le\frac{c-b}{b-a}\le\frac{1}{2}\]
\end{proof}

\begin{lemma}\label{th-bisection}
Fix a base~$B=\{0,1,+,-,\times,\ch,\dotsc\}$. Let $a<b$ be integers, and
let $f_t\colon[2^{2^a},2^{2^b}]\to\R$ be a family of continuous monotonic
functions parameterized by~$t\in T\subset\R^n$. Assume that the
family~$f_t$ is represented as a $B$-circuit---i.e.\ there is a $B$-circuit $f$ that takes inputs $t$ and~$x$ and computes~$f_t(x)$.
Moreover, assume that $f_t(2^{2^a})f_t(2^{2^b})<0$ for all~$t\in T$.
Then there is a $B$-circuit~$g$ representing a function from~$\R^n$ to~$\R$
mapping any~$t\in T$ to a power of~$2$ such that
\[
	f_t(g(t)) f_t(2g(t)) \le 0
\]
and $\size{g}\le p(\abs{a}+\abs{b}+\|f\|)$ for some fixed polynomial~$p$.
\end{lemma}
\begin{proof}
The lemma says that we can find the order of magnitude of the
solution~$x$ of the equation~$f_t(x)=0$ uniformly in~$t$
through a circuit of size polynomial in~$\abs{a}+\abs{b}+\size{f}$. It is
easy to devise a bisecting procedure that finds the binary digits of an
integer~$e_t$ in such a way that
\[
	f_t(2^{e_t}) f_t(2^{e_t+1}) \le 0
\]
and to turn such procedure into a $B$-circuit of the required size.
\end{proof}

\begin{proof}[Proof of Proposition~\ref{th-polynomial-roots}]
For $d=1$ the result follows immediately observing that $r_d(x,y)=-x/y$.
Hence, it suffices to produce a polynomial time algorithm that, for~$d>1$, given a
$B_d$-circuit~$c$, computes a $B_{d-1}$-circuit~$\tilde c$ such that
$\tilde c$ is positive if and only if $c$ is positive.
First we use
Lemma~\ref{th-regularization} to get a regular circuit~$\bar{c}=c$. Our plan,
now, is to approximate the $\r_d$~gates of~$\bar{c}$ with suitable
$B_{d-1}$-circuits, in
such a way that the errors introduced do not alter the sign of any gate
in~$\bar{c}$, and ultimately of $\bar{c}$~itself. We will henceforth
produce a new $B_{d-1}$-circuit~$\tilde{c}$ obtained from~$\bar{c}$
through the replacement of $\r_d$~gates by $B_{d-1}$-circuits based on the
Newton's approximation method. We need $\bar{c}$ to be regular because our
replacement circuits work just for regular gates,
which is enough by Observation~\ref{th-taint}.

Let $N^1_d(x_0,a_0\dotsc a_d)$ be a $B_{d-1}$-circuit representing one
iteration of Newton's method applied to the polynomial~$p(x)\eqdef\sum_i
a_ix^i$, i.e.
\[
N^1_d(x_0,a_0\dotsc a_d) = x_0 - \frac{\sum_i a_ix_0^i}{\sum_i a_iix_0^{i-1}}
\]
For every positive integer~$k$, build a $B_{d-1}$-circuit $N^k_d(x_0,a_0\dotsc a_d)$ representing the $k$-th
iterate of Newton's method applied to~$p$ starting from~$x_0$ through the
iteration
\[
N^{k+1}_d(x_0,a_0\dotsc a_d) = N^k_d(N^1_d(x_0,a_0\dotsc a_d),a_0\dotsc
a_d)
\]
it is clear that, for fixed~$d$, the size of~$N^k_d$ is linear in~$k$.
Let $\zeta$ be the largest real root of~$p$.
We will produce a $B_{d-1}$-circuit~$A_d(a_0\dotsc a_d)$ that, assuming
that $\r_d(a_0\dotsc a_d)$ is regular,
evaluates to an approximate zero of~$p$ associated to~$\zeta$, moreover
we will arrange that the additional
condition~$\abs{A_d(a_0\dotsc a_d)}\le1+2^{2^{C_d\size{\bar c}}}$
holds.
Assuming that we have~$A_d$, we build~$\tilde c$
substituting each~$\r_d(a_0\dotsc a_d)$ with
\[
\tilde\r_d(a_0\dotsc a_d)\eqdef
N^{(E_d+C_d+1)(\size{\bar{c}}+K)}_d(A_d(a_0\dotsc a_d),a_0\dotsc a_d)
\]
for a suitable positive integer~$K$ depending only on~$d$,
which we will describe later in the proof.
At each regular $\r_d$~gate, our substitution introduces an error
bounded by
\[
\epsilon_K\eqdef2^{-2^{E_d(\size{\bar{c}}+K)}}
\]
Therefore, irrespective of the value of~$K\ge1$,
this substitution satisfies the hypothesis of
Lemma~\ref{th-perturbations} with~$e=0.5$, and, by
Lemma~\ref{th-heights-bound}, we have that $\tilde c - 2^{-2^{C_d\size{\bar
c}}-1}$ is positive if and only if $\bar c$ is positive, hence the
statement.

The last remaining step is the construction of~$A_d$. The circuit~$A_d$
that we are going to build will depend on the size of~$\bar c$, i.e.\ it
is going to work only in our particular situation, and not
necessarily for any value of its inputs~$a_0\dotsc a_d$. To
follow the proof, it is convenient to consider the construction of~$\tilde
c$ as a stepwise process in which $\r_d$~gates are replaced in evaluation
order. At a given step we are going to replace~$\r_d(a_0\dotsc a_d)$
with~$\tilde\r_d(\tilde a_0\dotsc \tilde a_d)$, where
$\tilde a_0\dotsc\tilde a_d$ are $B_{d-1}$-circuits obtained
from~$a_0\dotsc a_d$ during the preceding steps of the construction,
which, by induction, we can assume to yield precision~$\epsilon_K$ at each
$\r_d$~gate.
Remind that if $\r_d(a_0\dotsc a_d)$ is not regular, then, by
Observation~\ref{th-taint}, there is nothing to
prove, therefore we can assume this gate to be regular. Let $\tilde p$
denote~$\sum_i \tilde a_ix^i$. The derivatives of $p$
at~$\zeta$ can be expressed as
$B_d$-circuits with inputs~$a_0\dotsc a_d$ of size bounded by some~$k$
depending only on~$d$---let $\r_d^i(a_0\dotsc a_d)$ denote the circuit for the
$i$-th derivative of~$p$ at~$\zeta$.
Therefore, for each~$i\le d=\operatorname{deg}(p)$,
Lemma~\ref{th-heights-bound} gives us
\[
2^{-2^{C_d(\size{\bar{c}}+k)}} < \abs{p^{(i)}(\zeta)}
< 2^{2^{C_d(\size{\bar{c}}+k)}}
\]
On the other hand, choosing $K>k$, we see that the derivatives of~$\tilde
p$ at its largest real root~$\tilde\zeta$ approximate the derivatives
of~$p$ at~$\zeta$ to an absolute error smaller than
\[
\frac{1}{2} 2^{-2^{C_d(\size{\bar{c}}+k)}}
\]
in fact, this bound follows applying Lemma~\ref{th-perturbations} to the
circuits obtained plugging $\tilde a_0\dotsc\tilde a_d$ into~$\r_d^i$.
Hence we have that, for each~$i\le d$
\[
2^{-2^{C_d(\size{\bar{c}}+k)}-1} < \abs{\tilde{p}^{(i)}(\tilde\zeta)}
< 2^{2^{C_d(\size{\bar{c}}+k)}+1} \tag{$\star$}
\]
and, in particular, none of the first $d$ derivatives of~$\tilde{p}$
vanishes at~$\tilde{\zeta}$.

By Tarski-Seidenberg, we can test the signs of the derivatives of~$\tilde
p$ at~$\tilde\zeta$ through a fixed Boolean combination of
polynomial conditions on the coefficients~$a_0\dotsc a_d$, and, in turn,
we can realize this Boolean combination as a switching network of
$\ch$~gates. Therefore we can build a circuit designed to decide its course
of action based on the signs of the derivatives of~$\tilde p$
at~$\tilde\zeta$.
From now on, we assume that the first and second derivatives of~$\tilde p$
at~$\tilde\zeta$ are positive: the reader can easily work out the three
other cases.

Our goal, now, is to find $a$, $b$, and~$c$, with $c$ represented by a
$B_{d-1}$-circuit, satisfying the hypothesis of
Lemma~\ref{th-approx-zero}: this is enough to conclude. First we use
$\r_{d-1}$ gates to write the roots of the first, second, and third
derivatives of~$\tilde p$. We call~$S$ the set of all these values plus~$\pm
(1+2^{2^{C_d\size{\bar c}}})$, which are an upper and lower bound
for~$\tilde\zeta$---to obtain this bound, similarly as we did for the derivatives,
we can apply the inductive hypothesis and
Lemma~\ref{th-perturbations} to $\r_d(\tilde a_0\dotsc\tilde a_d)$,
we deduce that $\tilde\zeta$ is very close to~$\zeta$, which, in turn, is less
than~$2^{2^{C_d\size{\bar c}}}$ in absolute value by Lemma~\ref{th-heights-bound}.
Now,
using choice gates, we find two consecutive elements $s$ and~$t$ of~$S$
such that~$s\le\zeta\le t$.
The polynomial~$\tilde{p}$ and its second
derivative~$\tilde{p}''$ are
monotonic in the interval~$[s,t]$, hence we can apply
Lemma~\ref{th-bisection} to the
composition~$p\circ p''^{-1}$ -- condition~\ref{cond-bisection} following from~($\star$) -- and get a new
interval~$[a,d]\subset[s,t]$ such that $\zeta\in[a,d]$ and there is an~$e$ such that
\[
2^{e} \le \tilde{p}''(x) \le 2^{e+1}
\]
for all~$x\in[a,d]$.

Now, in order to find $b$ and~$c$, we turn our attention to the order of magnitude of~$\tilde{\zeta}-a$.
First we use a choice gate to test whether~$\tilde\zeta$ happens to be
within~$\epsilon_K$ from~$a$. If this is the case, then we let $A_d$ just
output~$a+\epsilon_K$---the reader may notice that, strictly speaking, in
this case the output of~$A_d$
may not be an approximate zero, nevertheless it is already as close
to~$\tilde\zeta$ as we need the final result of the Newton's iterations to
be, so, for the purpose of our algorithm, no harm is done.
If $\tilde\zeta\ge a+\epsilon_K$, then
\[
2^{-2^{E_d(\size{\bar c}+K)}} \le \tilde\zeta -a \le 1+2^{2^{C_d\size{\bar
c}}}
\]
i.e.\ the order of magnitude of~$\tilde\zeta -a$ is bounded in a range of
size linear in~$\size{\bar c}$.
Again using Lemma~\ref{th-bisection}, we can find $b'$ and~$c'$ such
that $a<b'\le\tilde\zeta\le c'$ and $c'-a = 2(b'-a)$. Choosing either $b=b'$
and~$c=(b'+c')/2$, or $b=(b'+c')/2$ and~$c=c'$, by means of two last
choice gates, we get $b$ and~$c$ such that
$c-b \le (b-a)/2$ and $b\le\tilde\zeta\le c$. This concludes the proof.
\end{proof}

\section{Reducing \mathshift\times\mathshift\ to an Arbitrary Semialgebraic Function} \label{sect-times-to-funct}

\noindent In this section we address the opposite problem of
Section~\ref{sect-polynomial-roots}, namely we want to recover~$\times$
starting from a regular not piecewise linear function~$f\colon\R\to\R$.
The argument is comparatively technically easier. The idea is to observe
that the product can be simulated using linear operations and the square
function. In turn, the square can be approximated, in some sense, zooming
in a point on the graph of~$f$, because we can expect that, under strong
magnification, $f$ should be practically indistinguishable from its second
order approximation. The zero bound, in this direction, is trivial,
because the input circuit computes an integer value.
Technical obstacles lie in the fact that we can use
just~$+$ and~$-$, as opposed to all linear functions, and in balancing
the quality of our approximation with the size of the resulting circuit.

\begin{lemma}\label{th-times-to-poly}
Let $p\in\R[x]$ be a polynomial of degree $d>1$, and let $\alpha$ be a positive real.
Then $\posslp$ is polynomial time many-one reducible
to~$\posslp(0,\alpha,+,-,p)$.
\end{lemma}
\begin{proof}
We will produce a polynomial
time procedure that, given a closed circuit~$c$ with gates in the
basis~$B\eqdef\{0,1,+,-,\times\}$, generates a closed circuit~$\tilde c$ with gates
in~$B'\eqdef\{0,\alpha,+,-,p\}$, in such a way that $c$ is positive if and
only if $\tilde c$ is positive.

First we prove that without loss of generality we can assume $p(x)=ax^2$
for some real coefficient~$a>0$. Let $q$ be the polynomial
\[
q(x) = \Delta_\alpha^{d-2}[p](x)
\]
where $\Delta_\alpha^{d-2}$ is the $(d-2)$-th iterate of the
first difference operator
\[
\Delta_\alpha[f](x) = f(x+\alpha) - f(x)
\]
Clearly $q$ can be implemented with a fixed number, depending on~$d$, of
$B'$-gates, and the degree of~$q$ is~$2$. Now, the polynomial
\[ q(2x) - 2q(x) + q(0) \]
has the required form, except at most for the sign of~$a$, which is easily
corrected.

Assuming $p(x)=ax^2$, we can construct a $B'$-circuit computing
\[
2axy = p(x+y) - p(x) - p(y)
\]
therefore, from now on, we can replace $B'$ with
\[\{0,\alpha,+,-,(x,y)\mapsto2axy\}\]
We will need the following observation, that, for each~$n$, we can construct a
closed $B'$-circuit $k_n=(2a)^{2^n-1}\alpha^{2^n}$ of size linear in~$n$.
In fact, it suffices to let $k_0=\alpha$ and~$k_{i+1}=2ak_ik_i$, where in
the latter a $2axy$~gate is applied to a single instance of~$k_i$ as $x$
and~$y$.

Now, we construct an intermediate $B$-circuit $c'$ in such a way that
$c'=c$, the circuit $c'$ has only one $1$~gate, and all the paths from any given gate to
the output gate have the same length---this can be accomplished adding
to~$c$ no more than~$\size{c}^2$ dummy gates arranged in $0+0+\dotsb$
subcircuits.
We may assume the depth of a non-constant gate of~$c'$ to be its distance
from the $1$~gate---in fact, if the $1$~gate is not above some gate~$g$,
then the value of~$g$ is necessarily~$0$, and we can replace $g$ with the
constant~$0$.
The circuit~$c'$ is our template for the construction of~$\tilde
c$---the $0$~gates of~$c'$ correspond in~$\tilde c$ to $0$~gates, the only
$1$~gate corresponds to an $\alpha$~gate, for every $xy$
gate in $c'$ we place a corresponding $2axy$ gate in $\tilde
c$, and for every $x\pm y$
gate, occurring, say, at depth~$d$, we place a $2ak_{d-1}(x\pm y)$ subcircuit. 
It is easy to show, by induction on the depth, that the value of a gate at
depth~$d$ of $c'$ multiplied by  $(2a)^{2^{d}-1}\alpha^{2^{d}}$ gives
the value of the
corresponding object in~$\tilde c$. Therefore $\tilde c$ and~$c=c'$ have the same sign.
\end{proof}

\begin{proposition}\label{th-root-to-square}
Let $B$ be the
basis~$\{0,+,-,f,\dotsc\}$ where $f$ denotes a unary function
and the dots -- $\dotsc$ -- indicate additional functions.
Assume that for some $\alpha,\beta>0$ we have
\[
\abs{f(x)-\alpha x^2}\le \abs{x}^3
\]
for all $x\in[-\beta,\beta]$, and assume that $V(B)$ is dense in~$\R$.
Then $\posslp$ is polynomial time many-one reducible to~$\posslp(B)$.
\end{proposition}
\begin{proof}
By Lemma~\ref{th-times-to-poly} suffices to show
a polynomial time procedure that, given a circuit~$c$ with
gates in the basis $B'\eqdef\{0,1,+,-,\cdot^2\}$,
generates a $B$-circuit~$\tilde c$ in such a way
that $c$ is positive if and only if $\tilde c$ is positive.

Without loss of generality we may assume $2\beta\le\alpha$.
First we describe a procedure that, given positive numbers
$e\le\alpha$ and~$u$, with $u$ represented as a $B$-circuit, and given a
$B'$-circuit~$n$, attempts to produce a
$B$-circuit $a(n,u,e)$ such that
\begin{equation}
\abs{a(n,u,e)-nu}<\frac{eu}{\alpha} \tag{$\star$}
\end{equation}
the intuition is that
$a(n,u,e)$ represents an $e/\alpha$-approximation of~$n$ scaled
down by a factor of~$u$. This procedure may either succeed in its goal or
fail explicitly.

The circuit~$a(n,u,e)$ is defined inductively on the structure
of~$n$
\begin{align*}
a(m^2,f(u),e) &=
f(a(m,u,\frac{e-4(|m|+1)^3u}{4|m|+2}))
\intertext{\hfill fail if $e\le4(|m|+1)^3u$ or $\beta<(|m|+1)u$}
a(m_1\pm m_2,u,e) &= a(m_1,u,\frac{e}{2})\pm a(m_2,u,\frac{e}{2}) \\
a(1,u,e) &= u \\
a(0,u,e) &= 0
\end{align*}
the procedure fails in any other case. It is easy to check by direct
computation that
property~($\star$) is preserved---for the first case, the computation goes
as follows. Assume $(|m|+1)u\le\beta$ and remember that $e\le\alpha$ and
$m$ is an integer. Let
\[\epsilon = \frac{e-4(|m|+1)^3u}{4|m|+2}\]
then clearly $0<\epsilon\le\alpha$, and we have
\begin{multline*}
\abs{f(a(m,u,\epsilon))-m^2f(u)} \le\\
\begin{split}
&{\le\;}\phantom{+\;}\abs{f(a(m,u,\epsilon))-\alpha(a(m,u,\epsilon))^2}\\
&\phantom{\le\;}{+\;}\abs{\alpha(a(m,u,\epsilon))^2 - \alpha m^2u^2}\\
&\phantom{\le\;}{+\;}\abs{\alpha m^2u^2 - m^2f(u)} \\
&{\le\;} (|m|+1)^3u^3+\epsilon(2|m|+1)u^2+|m|^2u^3 \\
&{<\;} \frac{eu^2}{2} \le \frac{ef(u)}\alpha
\end{split}
\end{multline*}
hence property~($\star$) is preserved.

Now we find $u_n$ such that $a(n,u_n,\alpha)$ does not fail.
To this aim, we use the density hypothesis to pick a $B$-circuit~$k$ such that
\[
0< k\le \operatorname{min}(\frac{1}{3\alpha}\,,\frac{\beta}{2})
\]
Observe that $0< f^i(k) \le \beta 2^{-2^i}$ for all~$i$, this can be shown
proving by induction that $0<f^i(k)\le k 2^{1-2^i}$.
We claim that
$u_n\eqdef f^{3\size{n}+3}(k)$ does the job.
In fact, at each step~$a(n',u',e')$
of the recursion, the following conditions are met
\begin{align*}
u' &\le \beta 2^{-2^{3\size{n}+3-d}}\\
e' &\ge \alpha 2^{-d2^{\size{n}+3}}
\end{align*}
where $d$ is the depth at which the recursion step occurs---the first one
follows from the bound on~$f^i(k)$, the second can be shown by induction
using the first plus~$n'<2^{2^{\size{n}}}$. From this conditions,
straightforward computation shows that the procedure does not fail.

For the special case of computing a circuit representation of
$a(n)\eqdef a(n,u_n,\alpha)$, we argue that a variant
of our procedure can be carried out in polynomial time. First, we
already know that the procedure will not fail, hence we can omit to
maintain the value of~$e$, since this quantity is used uniquely to check
for failure.
Then, notice that the choice of~$k$ does not
depend on the circuit~$n$ (but just on the basis~$B$), therefore we can simply pick a valid~$k$ and hard-code it into the procedure.
Hence, since the only values of~$u$ that we encounter
during the performance of the procedure are of the form~$f^{i+\size{n}+4}(k)$
with $0\le i<\size{n}$,
we can construct $B$-circuits to represent these values
in time polynomial in~$\size{n}$. Now, stipulate that every
time we need to compute~$a(n',u',-)$ for a subcircuit~$n'$ of~$n$ and
a~$u'$ in our list, we check if the same computation has already been
performed, and if so we simply link to the already constructed subcircuit.
Since there are $\size{n}$ possible subcircuits~$n'$, and $u'$ ranges over a set
of~$\size{n}$ different values, our modified procedure makes at most
$\size{n}^2$ recursive calls.

Finally, property~($\star$) yields
\[
\abs{\frac{a(2c-1)}{u_{2c-1}} - (2c-1)}< 1
\]
and, observing that $2c-1$ necessarily represents an odd integer,
this implies that $\tilde c\eqdef a(2c-1)$
is positive if and only if $c$ is positive.
\end{proof}

\section{Proof of the Main Results}

\begin{proof}[Proof of Theorem~\ref{th-main-a}]
It is easy to see, for instance by the Finiteness
Theorem~\cite[Theorem~2.4.1]{PreDel}, that each function in~$B$ is
piecewise algebraic, with finitely many semialgebraic pieces.
Hence, there is a degree~$d$ such that each function in~$B$
can be written as a $B_d$-circuit, where~$B_d$ denotes the language
defined at the beginning of Section~\ref{sect-polynomial-roots}. Therefore
$\posslp(B)$ is reducible to~$\posslp(B_d)$, which, in turn, is reducible
to~$\posslp$ by Proposition~\ref{th-polynomial-roots}.
\end{proof}

\begin{proof}[Proof of Theorem~\ref{th-main-b}]
First assume that there is function~$g\colon\R^n\to\R$ in~$B$ that
is not piecewise linear. By o-minimality (see e.g.\ \cite[Chapter~8
Exercises~3.3]{VanDenDries}), $g$ is of class~$\mathrm{C}^2$ in an open
subset~$O$ of~$\R^n$ of codimension~$<n$. Since $O$ is dense, $g|_O$ can
not be linear, hence we can pick a point~$x\in O\cap V(B)$
such that the Hessian matrix~$H(g)(x)$ of~$g$ at~$x$ does not vanish.
Now we pick an integer vector~$v\in\Z^n$ such that $v^{\textrm{\sc
T}}H(g)(x_0)v$ does not vanish. Clearly
\begin{align*}
	f\colon\R&\to\R\\
	t&\mapsto 2g(x+tv) - g(x+2tv) - g(x)
\end{align*}
can be represented by a $B$-circuit, and it is of class~$\mathrm{C}^2$
at~$0$. It can be verified by direct computation that $f(0)=f'(0)=0$
and~$f''(0)\neq 0$. Hence either $f$ or~$-f$ satisfies the hypothesis of
Proposition~\ref{th-root-to-square}. Therefore we have the first case.

Now assume that all functions in~$B$ are piecewise linear, albeit possibly
with algebraic coefficients. We will show how to evaluate $B$-circuits in
polynomial time. Fix a number field~$K$ in which all the coefficients of
all the linear pieces of functions in~$B$ reside. Fix~$e_1\dotsc e_n\in K$
that generate~$K$ as a vector space over~$\Q$. Clearly $V(B)\subset K$,
hence, in the evaluation of $B$-circuits, we can restrict the domain of
our computation to~$K$. We represent each element~$x$ of~$K$ using the
unique vector~$\sigma(x)\in\Q^n$ such that~$x=\sum_i \sigma(x)_ie_i$. Now,
for each function~$f\in B$ we need to know how to
compute~$\sigma(f(x_1\dotsc x_m))$ given~$\sigma(x_1)\dotsc\sigma(x_m)$.
Let $g(x_1\dotsc x_m)$ be one of the linear pieces that constitute~$f$,
then $\sigma\circ g\circ\sigma^{-1}$ is a linear map from~$(\Q^n)^m$
to~$\Q^n$, therefore we decide to compute the linear pieces
constituting~$f$ simply by matrix multiplication in the rationals, kept as
pairs of coprime integers, which in turn are kept in binary.
To choose among the pieces, by
Tarski-Seidenberg, suffices to evaluate a fixed (depending on~$f$) set of
rational polynomial conditions on the coefficients
of~$\sigma(x_1)\dotsc\sigma(x_m)$, which we decide to do again by simple
rational arithmetic. Finally, to decide~$x>0$ given~$\sigma(x)$, we employ
similarly Tarski-Seidenberg to translate this condition into a Boolean
combination of polynomial conditions on~$\sigma(x)_1\dotsc\sigma(x)_n$, and
we evaluate it using rational arithmetic. Summarizing, we precompute the
coefficients for the finite number of rational linear functions and
polynomials that we will need, this data depends only on~$B$, which is
fixed. Then we carry out the evaluation as described above. It is easy to
check that the algorithm works in polynomial time. This concludes the
proof of the second case.
\end{proof}

\section{Addenda}
\label{sect-addenda}

\noindent In this last section we collect a few additional results, observations, questions.

First we would like to discuss
the density and continuity assumptions in
Theorem~\ref{th-main-b}.
If all the functions in~$B$ happen to be constants or unary
functions, we can dispense with the aforementioned assumptions because of
the following fact, which is a consequence of~\cite[Corollary~2.2]{Wil}
and Fact~\ref{th-poly-bound}.

\begin{fact}\label{th-discrete-poly}
Let $f\colon\R\to\R$ be a semialgebraic function. Assume that $f$ maps
integers to integers. Then there is~$N\in\R$ such that $f|_{[N,+\infty[}$
coincides with a polynomial.
\end{fact}

\begin{theorem}
Let $B=\{+,-,f_1\dotsc f_n\}$ be a finite set of functions semialgebraic over~$\Q$. Assume that $f_1\dotsc f_n$ are either constants or unary
functions. Then the following dichotomy holds:
$\posslp(B)$ is in~$\mathsf P$ if all the functions in~$B$ are piecewise linear when
restricted to~$V(B)$,
otherwise $\posslp$ and~$\posslp(B)$ are mutually
polynomial time Turing reducible.
\end{theorem}
\begin{proof}
Along the lines of the proof of Theorem~\ref{th-main-b}.
The only new case
to examine is when~$V(B)$ is discrete and there is a function~$f_i$ that
restricted to~$V(B)$ is not piecewise linear. Since $V(B)$ must be an
additive subgroup of~$\R$, because $+,-\in B$, we can assume that $V(B)$ is
precisely~$\Z$.
By Fact~\ref{th-discrete-poly}, we have that $f_i$ must
coincide with a non-affine polynomial over the integers of a half line.
Our goal, now, is to use $f_i$ in order to construct a $B$-circuit
representing a 
function~$g\colon\R\to\R$ that coincides
with a non-affine polynomial~$\tilde g$ on all of~$\Z$---as opposed to
just a half line.
Assuming that we succeed in this, then we can immediately conclude.
In fact,
the basis~$(0,1,+,-,\tilde g)$ satisfies the
hypothesis of Lemma~\ref{th-times-to-poly}, and testing $(0,1,+,-,\tilde
g)$-circuits is equivalent to testing $(0,1,+,-,g)$-circuits, since,
by our assumption that $V(B)=\Z$, any
closed $(0,1,+,-,\tilde g)$-circuit
has the same value of the $(0,1,+,-,g)$-circuit obtained replacing
$\tilde g$~gates with $g$~gates.

The function~$g$ is constructed as follows.
By Fact~\ref{th-discrete-poly}, the function~$f_i$ coincides with a polynomial~$p$
on~$[N,+\infty[$, and with another polynomial~$q$ on~$]-\infty,-N]$, for
some positive integer~$N$.
Let $d$ be the maximum of~$\operatorname{deg}(p)$
and~$\operatorname{deg}(q)$. Consider
\begin{align*}
h_1&\eqdef\Delta_1^{d-2}[f_i] & h_2(x)&\eqdef h_1(2x)-2h_1(x)
\end{align*}
As in the proof of Lemma~\ref{th-times-to-poly}, it is easy to see that
there are $a$, $b$, $a'$, and~$b'$, with at least one of~$a$ and~$a'$ not
null, such that
\begin{align*}
h_2(x) &= a x^2 + b &&\text{for $x\in{]}-\infty,-M{]}$} \\
h_2(x) &= a' x^2 + b' &&\text{for $x\in{[}M,+\infty{[}$}
\end{align*}
where $M=N+d-2$. If $a+a'\neq0$, then
\[
h_3(x) \eqdef h_2(x) + h_2(-x)
\]
coincides with the second degree polynomial
\[(a+a')x^2+(b+b')\]
for~$\abs x \ge M$.
Therefore, if~$a+a'\neq 0$, we can simply choose $g(x) = h_3(M(2x+1))$. It remains to
be considered the case~$a=-a'$. In this case we observe that, for
$\abs{x}\ge M+1$
\[
	h_2(x+1) - h_2(x-1) = 4a\abs{x}
\]
Therefore we can replace $h_2$ in the above argument with
\[
\bar{h}_2(x) \eqdef h_2(h_2(x+1)-h_2(x-1))
\]
which coincides with~$4a^2x^2+b'^{\text{-or-not}}$ for $\abs{x}$ sufficiently large.
\end{proof}

Unfortunately, we do not have an analogue of
Fact~\ref{th-discrete-poly} for multivariate functions. Moreover,
the following example shows that, failing either the continuity
or the density hypothesis, we can not employ the technique of reducing to
a unary function and invoke Lemma~\ref{th-times-to-poly}
or Proposition~\ref{th-root-to-square}. In particular $\posslp(B)$ may be equivalent to~$\posslp$
even though all the unary functions represented by $B$-circuits are either
constant on a cofinite set, or the identity function.

\begin{example}
Let
\begin{align*}
B &=\{c_1\dotsc c_n,+,-,\times\} \\
B' &=\{c_1\dotsc c_n,c_1^2\dotsc c_n^2,g_+,h_+,g_-,h_-,g_\times,h_\times\}
\end{align*}
where for $f\colon\R^2\to\R$ we define
\begin{align*}
g_f\colon\R^4&\to\R \\
{(x,y,z,t)} &\mapsto \begin{cases}
f(x,y) &\text{if $z=x^2$ and $t=y^2$} \\
0 &\text{otherwise}
\end{cases}\\
h_f\colon\R^4&\to\R \\
{(x,y,z,t)} &\mapsto \begin{cases}
{(f(x,y))^2} &\text{if $z=x^2$ and $t=y^2$} \\
0 &\text{otherwise}
\end{cases}
\end{align*}

Clearly $\posslp(B')$ is polynomial time Turing reducible
to~$\posslp(B)$. And the converse is also true. In fact, we can simulate 
any $B$-circuit~$c$ through a $B'$-circuit that keeps for each gate~$g$
of~$c$ a pair of gates, $g_1$ and~$g_2$, computing $c$ and~$c^2$
respectively.

On the other hand, for any function~$f\colon\R^4\to\R$ in~$B'$, and for any choice of four
linear functions $a_1\dotsc a_4\colon\R\to\R$, we see that
$f(a_1(x)\dotsc a_4(x))$, as a function of~$x$, is constant on a
cofinite set. It follows that any unary function represented by a
$B'$-circuit, unless it is the function~$x\mapsto x$, must be constant
outside of a finite set.
\end{example}

In opposition to continuity and density,
we are unsatisfied by the hypothesis that $B$ contains~$+$ and~$-$
in Theorem~\ref{th-main-b},
and would like to see it weakened or eliminated.
This hypothesis comes directly
from Proposition~\ref{th-root-to-square}, where we need $+$ and~$-$ to
manipulate the graph of~$f$.
It is conceivable that, as we can simulate~$\times$ killing the constant
and first degree terms of a second
order approximation of~$f$, we may be able to simulate some linear
function killing the constant and second degree terms, at least if $f$ is
generic enough.

The initial motivation of the present work
has been an ongoing attempt by the author and
Manuel Bodirsky to
investigate constraint satisfaction problems over the reals, continuing
the work initiated by~\cite{BJO12}. For technical reasons, due to the convexity
requirement proven in~\cite{BJO12}, it would be desirable to assess the
computational complexity of the problem of comparing two $B$-circuits (as
opposed to one $B$-circuit and~$0$) in a basis $B$ not containing the~$-$
function. For
the case of the basis~$\{0,1,+,-,\times\}$, it is an observation that the comparison of
circuits on the basis~$\{0,1,+,\times\}$ is polynomial time equivalent
to~$\posslp$. Nevertheless, we can not say whether the same holds
in a more general situation.

\section*{Acknowledgement}

\noindent The author is indebted to Manuel Bodirsky for pointing out the direction
that led to the present work, for numerous discussions and suggestions, and
for commenting on preliminary versions. We would also like to express our
gratitude to Eleonora Bardelli for interesting discussions and comments,
and to Vincenzo Mantova for suggesting references on Weil heights.
A number of improvements to this work would not have been possible without
valuable feedback provided by the anonymous reviewers of the
\textsc{csl-lics} joint meeting to be held in Vienna, July 14--18, 2014.



%


\begin{thebibliography}{ABKM09}
\providecommand{\url}[1]{#1}
\csname url@samestyle\endcsname
\providecommand{\newblock}{\relax}
\providecommand{\bibinfo}[2]{#2}
\providecommand{\BIBentrySTDinterwordspacing}{\spaceskip=0pt\relax}
\providecommand{\BIBentryALTinterwordstretchfactor}{4}
\providecommand{\BIBentryALTinterwordspacing}{\spaceskip=\fontdimen2\font plus
\BIBentryALTinterwordstretchfactor\fontdimen3\font minus
  \fontdimen4\font\relax}
\providecommand{\BIBforeignlanguage}[2]{{%
\expandafter\ifx\csname l@#1\endcsname\relax
\typeout{** WARNING: IEEEtranSA.bst: No hyphenation pattern has been}%
\typeout{** loaded for the language `#1'. Using the pattern for}%
\typeout{** the default language instead.}%
\else
\language=\csname l@#1\endcsname
\fi
#2}}
\providecommand{\BIBdecl}{\relax}
\BIBdecl

\bibitem[ABKM09]{ABKM08}
\BIBentryALTinterwordspacing
E.~Allender, P.~B{\"u}rgisser, J.~Kjeldgaard-Pedersen, and P.~B. Miltersen,
  ``On the complexity of numerical analysis,'' \emph{SIAM J. Comput.}, vol.~38,
  no.~5, pp. 1987--2006, 2008/09. [Online]. Available:
  \url{http://dx.doi.org/10.1137/070697926}
\BIBentrySTDinterwordspacing

\bibitem[Atk89]{Atkinson}
K.~E. Atkinson, \emph{An introduction to numerical analysis}, 2nd~ed.\hskip 1em
  plus 0.5em minus 0.4em\relax New York: John Wiley \& Sons Inc., 1989.

\bibitem[BC06]{BraCo06}
M.~Braverman and S.~Cook, ``Computing over the reals: foundations for
  scientific computing,'' \emph{Notices Amer. Math. Soc.}, vol.~53, no.~3, pp.
  318--329, 2006.

\bibitem[BCSS98]{BCSS}
L.~Blum, F.~Cucker, M.~Shub, and S.~Smale, \emph{Complexity and real
  computation}.\hskip 1em plus 0.5em minus 0.4em\relax New York:
  Springer-Verlag, 1998, with a foreword by Richard M. Karp.

\bibitem[BJO12]{BJO12}
\BIBentryALTinterwordspacing
M.~Bodirsky, P.~Jonsson, and T.~v. Oertzen, ``Essential convexity and
  complexity of semi-algebraic constraints,'' \emph{Log. Methods Comput. Sci.},
  vol.~8, no.~4, pp. 4:5, 25, 2012. [Online]. Available:
  \url{http://dx.doi.org/10.2168/LMCS-8(4:5)2012}
\BIBentrySTDinterwordspacing

\bibitem[Bra05]{Bra05}
M.~Braverman, ``On the complexity of real functions,'' in \emph{Foundations of
  Computer Science, 2005. FOCS 2005. 46th Annual IEEE Symposium on}, 2005, pp.
  155--164.

\bibitem[Ers81]{Ersh81}
A.~P. Ershov, ``Abstract computability on algebraic structures,'' in
  \emph{Algorithms in modern mathematics and computer science ({U}rgench,
  1979)}, ser. Lecture Notes in Comput. Sci.\hskip 1em plus 0.5em minus
  0.4em\relax Berlin: Springer, 1981, vol. 122, pp. 397--420.

\newpage

\bibitem[EY10]{EtYa10}
\BIBentryALTinterwordspacing
K.~Etessami and M.~Yannakakis, ``On the complexity of {N}ash equilibria and
  other fixed points,'' \emph{SIAM J. Comput.}, vol.~39, no.~6, pp. 2531--2597,
  2010. [Online]. Available: \url{http://dx.doi.org/10.1137/080720826}
\BIBentrySTDinterwordspacing

\bibitem[FM92]{FrieMan92}
\BIBentryALTinterwordspacing
H.~Friedman and R.~Mansfield, ``Algorithmic procedures,'' \emph{Trans. Amer.
  Math. Soc.}, vol. 332, no.~1, pp. 297--312, 1992. [Online]. Available:
  \url{http://dx.doi.org/10.2307/2154033}
\BIBentrySTDinterwordspacing

\bibitem[GGJ76]{GGJ76}
M.~R. Garey, R.~L. Graham, and D.~S. Johnson, ``Some {NP}-complete geometric
  problems,'' in \emph{Eighth {A}nnual {ACM} {S}ymposium on {T}heory of
  {C}omputing ({H}ershey, {P}a., 1976)}.\hskip 1em plus 0.5em minus 0.4em\relax
  Assoc. Comput. Mach., New York, 1976, pp. 10--22.

\bibitem[Goo94]{Good94}
\BIBentryALTinterwordspacing
J.~B. Goode, ``Accessible telephone directories,'' \emph{J. Symbolic Logic},
  vol.~59, no.~1, pp. 92--105, 1994. [Online]. Available:
  \url{http://dx.doi.org/10.2307/2275252}
\BIBentrySTDinterwordspacing

\bibitem[Koi93]{Kor93}
\BIBentryALTinterwordspacing
P.~Koiran, ``A weak version of the {B}lum, {S}hub \& {S}male model,'' in
  \emph{34th {A}nnual {S}ymposium on {F}oundations of {C}omputer {S}cience
  ({P}alo {A}lto, {CA}, 1993)}.\hskip 1em plus 0.5em minus 0.4em\relax Los
  Alamitos, CA: IEEE Comput. Soc. Press, 1993, pp. 486--495. [Online].
  Available: \url{http://dx.doi.org/10.1109/SFCS.1993.366838}
\BIBentrySTDinterwordspacing

\bibitem[Lan83]{Lang}
S.~Lang, \emph{Fundamentals of {D}iophantine geometry}.\hskip 1em plus 0.5em
  minus 0.4em\relax New York: Springer-Verlag, 1983.

\bibitem[LPY05]{LPY05}
\BIBentryALTinterwordspacing
C.~Li, S.~Pion, and C.~K. Yap, ``Recent progress in exact geometric
  computation,'' \emph{J. Log. Algebr. Program.}, vol.~64, no.~1, pp. 85--111,
  2005. [Online]. Available: \url{http://dx.doi.org/10.1016/j.jlap.2004.07.006}
\BIBentrySTDinterwordspacing

\bibitem[PD01]{PreDel}
A.~Prestel and C.~N. Delzell, \emph{Positive polynomials}, ser. Springer
  Monographs in Mathematics.\hskip 1em plus 0.5em minus 0.4em\relax Berlin:
  Springer-Verlag, 2001, from Hilbert's 17th problem to real algebra.

\bibitem[Poi95]{Poizat}
B.~Poizat, \emph{Les petits cailloux}, ser. Nur al-Mantiq wal-Ma\rasp rifah
  [Light of Logic and Knowledge], 3.\hskip 1em plus 0.5em minus 0.4em\relax
  Lyon: Al\'eas, 1995, une approche mod{\`e}le-th{\'e}orique de l'algorithmie.
  [A model-theoretic approach to algorithms].

\bibitem[Sil86]{Silverman}
J.~H. Silverman, \emph{The arithmetic of elliptic curves}, ser. Graduate Texts
  in Mathematics.\hskip 1em plus 0.5em minus 0.4em\relax New York:
  Springer-Verlag, 1986, vol. 106.

\bibitem[Tiw92]{Tiwa92}
\BIBentryALTinterwordspacing
P.~Tiwari, ``A problem that is easier to solve on the unit-cost algebraic
  {RAM},'' \emph{J. Complexity}, vol.~8, no.~4, pp. 393--397, 1992. [Online].
  Available: \url{http://dx.doi.org/10.1016/0885-064X(92)90003-T}
\BIBentrySTDinterwordspacing

\bibitem[TV08]{TarVya07}
\BIBentryALTinterwordspacing
S.~P. Tarasov and M.~N. Vyalyi, ``Semidefinite programming and arithmetic
  circuit evaluation,'' \emph{Discrete Appl. Math.}, vol. 156, no.~11, pp.
  2070--2078, 2008. [Online]. Available:
  \url{http://dx.doi.org/10.1016/j.dam.2007.04.023}
\BIBentrySTDinterwordspacing

\bibitem[vdD98]{VanDenDries}
\BIBentryALTinterwordspacing
L.~van~den Dries, \emph{Tame topology and o-minimal structures}, ser. London
  Mathematical Society Lecture Note Series.\hskip 1em plus 0.5em minus
  0.4em\relax Cambridge: Cambridge University Press, 1998, vol. 248. [Online].
  Available: \url{http://dx.doi.org/10.1017/CBO9780511525919}
\BIBentrySTDinterwordspacing

\bibitem[Wal00]{Waldschmidt}
M.~Waldschmidt, \emph{Diophantine approximation on linear algebraic groups},
  ser. Grundlehren der Mathematischen Wissenschaften [Fundamental Principles of
  Mathematical Sciences].\hskip 1em plus 0.5em minus 0.4em\relax Berlin:
  Springer-Verlag, 2000, vol. 326, transcendence properties of the exponential
  function in several variables.

\bibitem[Wil04]{Wil}
\BIBentryALTinterwordspacing
A.~J. Wilkie, ``Diophantine properties of sets definable in o-minimal
  structures,'' \emph{J. Symbolic Logic}, vol.~69, no.~3, pp. 851--861, 2004.
  [Online]. Available: \url{http://dx.doi.org/10.2178/jsl/1096901771}
\BIBentrySTDinterwordspacing

\end{thebibliography}

\providecommand{\etalchar}[1]{$^{#1}$}
\def\rasp{\leavevmode\raise.45ex\hbox{$\rhook$}}

\end{document}